%% file: ICASSP2019.tex
\documentclass{article}
\usepackage{spconf,amsmath,graphicx,amssymb,algorithm,array,bm,optidef,cite,float}
\usepackage[algo2e]{algorithm2e}
\usepackage[utf8]{inputenc}
\usepackage[english]{babel}
\usepackage{amsthm}
\newtheorem{theorem}{Theorem}
\usepackage{xcolor}  

\title{A Fair and Scalable Power Control Scheme in Multi-cell Massive MIMO}
\name{ Amin Ghazanfari$^{\dagger}$, Hei Victor Cheng$^{\star}$, Emil Bj{\"o}rnson$^{\dagger}$, and Erik G. Larsson$^{\dagger}$\thanks{This paper was supported by ELLIIT and by the European Union’s Horizon 2020 research and innovation programme under grant agreement No 641985 (5Gwireless).}}

\address{$^{\dagger}$ Department of Electrical Engineering (ISY), Link{\"o}ping University, Sweden \\
	$^{\star}$ Department of Electrical and Computer Engineering, University of Toronto\\Email:\{amin.ghazanfari, emil.bjornson, erik.g.larsson,\}@liu.se, hvc@ieee.org }

\input{commands}
\begin{document}
	\ninept
	\maketitle
	\begin{abstract}
		This paper studies the transmit power optimization in a multi-cell massive multiple-input multiple-output (MIMO) system. To overcome the scalability issue of network-wide max-min fairness (NW-MMF), we propose a novel power control (PC) scheme. This scheme maximizes the geometric mean (GM) of the per-cell max-min spectral efficiency (SE). To solve this new optimization problem, we prove that it can be rewritten in a convex form and then solved using standard tools. To provide a fair comparison with the available utility functions in the literature, we solve the network-wide proportional fairness (NW-PF) PC as well. The NW-PF focuses on maximizing the sum SE, thereby ignoring fairness, but gives some extra attention to the weakest users. The simulation results highlight the benefits of our model which is balancing between NW-PF and NW-MMF.  
	\end{abstract}
	\section{Introduction}
	\label{sec:intro}
	Massive MIMO \cite{marzetta2010noncooperative} is a key technology in 5G \cite{Parkvall2017a,WWB5G}. It refers to a system in which the cellular base stations (BSs) are equipped with very many antennas. Massive MIMO supports spatial multiplexing of many users, beamforming, and spatial interference mitigation. It enhances the spectral and energy efficiency compared with conventional MIMO setups. Unlike conventional cellular systems, the power control (PC) in massive MIMO systems benefits from channel hardening, namely that the small-scale fading average out when having many antennas \cite{bjornson2017massive}. It means that in massive MIMO, one can optimize the data transmission power based on the  large-scale fading coefficients only, instead of optimizing with respect to the small-scale fading coefficients, which would require very rapid PC updates. PC schemes that maximize different utility functions have been considered in the literature \cite{bjornson2017massive,guo2014uplink,kammoun2014low,zhao2013energy,wu2016asymptotically,guo2016security,7122928,liu2017pilot1}. Max-min fairness (MMF) is one of the classical utility functions and it is studied for different setups in \cite{van2016joint,cheng2017optimal,van2017joint,yang2017massive,xiang2014massive}. It provides the same quality of service at all user locations, which is a highly desirable feature in future systems.

	Applying the MMF utility to a multi-cell massive MIMO network leads to a scalability issue. This is due to the fact that when increasing the number of cells and active users in the network, the probability of having a user with an extremely poor channel gets higher due to shadow fading. Therefore, network-wide MMF (NW-MMF) optimization leads to the situation in which all users in the network suffer from the weak channel of the worst user. Consequently, all users get low spectral efficiency (SE). Increasing the number of cells to infinity will eventually result in zero SE for all users in the network. This is a major problem that was pointed out in \cite{redbook,bjornson2017massive}, but seldom discussed in the academic literature where the simulation setups are often too small to observe overly small SEs.
Nevertheless, the NW-MMF schemes proposed in the literature are unsuitable for providing fairness in cellular networks.

	\subsection{Related works}
	
	A heuristic approach to resolve the scalability issue of NW-MMF was considered in \cite[Ch.~6]{redbook}. The idea is to maximize the minimum SE within each cell and then balance these values across cells. Hence, the weak users have a lower impact on the whole network performance and mostly affect their own cells. 
	The proposed algorithm in \cite{redbook} is computationally efficient, but relies on approximations and there is no guarantee of optimality.
	Inspired by this algorithm, we are proposing a new utility function that can be optimized rigorously: maximization of the geometric mean (GM) of the max-min SEs in each of the cells. The network-wide proportional fairness (NW-PF) utility was considered in \cite{bjornson2017massive} to balance between sum SE optimization and fairness. In simulations, it outperforms NW-MMF in terms of fairness for most users, but it gives no fairness guarantees except for giving non-zero SE to every user.

	\subsection{Contributions}
	\begin{itemize}
		\item We propose a novel PC scheme that solves the GM per-cell MMF problem. We then reformulate the problem to reach a convex formulation that can be solved to global optimality in an efficient way. 
The new scheme outperforms the heuristic scheme in \cite{redbook} in some cases and gives a comparable performance in other cases.
		\item To further investigate the benefits of the proposed PC scheme, we define and solve two more power control schemes for the multi-cell scenario at hand: NW-MMF and NW-PF. The numerical results show that the proposed PC scheme combines the benefits of NW-MMF and NW-PF without suffering from the scalability issue of NW-MMF.
	\end{itemize}
\begin{figure*}
	{{\begin{equation}\tag{4}\label{eq:ul_sinr}
			\mathrm{SINR}^{\rm ul}_{lk}=\frac{M \rho_{\rm ul} \gamma^{l}_{lk}\eta_{lk}}{1+\rho_{\rm ul} \sum\limits_{l'\in \mathcal{P}_l}^{} \sum\limits_{k'=1}^{K} \beta^{l}_{l'k'}\eta_{l'k'} + \rho_{\rm ul}\sum\limits_{l'\notin \mathcal{P}_l}^{}\sum\limits_{k'=1}^{K}\beta^{l}_{l'k'}\eta_{l'k'}+M\rho_{\rm ul}\sum\limits_{l'\in \mathcal{P}_l\backslash\{l\}}{}\gamma^{l}_{l'k}\eta_{l'k}}
			\end{equation}}
		\hrulefill
}\end{figure*}
\begin{figure*}
	{{\begin{equation}\tag{5}\label{eq:dl_sinr}
			\mathrm{SINR}^{\rm dl}_{lk}=\frac{M \rho_{\rm dl} \gamma^{l}_{lk}\eta_{lk}}{1+\rho_{\rm dl} \sum\limits_{l'\in \mathcal{P}_l}^{}\beta^{l'}_{lk} \left(\sum\limits_{k'=1}^{K} \eta_{l'k'}\right) + \rho_{\rm dl}\sum\limits_{l'\notin \mathcal{P}_l}^{}\beta^{l'}_{lk} \left(\sum\limits_{k'=1}^{K}\eta_{l'k'}\right)+M\rho_{\rm dl}\sum\limits_{l'\in \mathcal{P}_l\backslash\{l\}}^{}\gamma^{l'}_{lk}\eta_{l'k}}
			\end{equation}}
		\hrulefill
}\end{figure*}
	\section{System Model} 
	\label{SystemModel}
	In this paper, we consider a multi-cell massive MIMO setup that consists of $L$ cells, each associated with one BS. Each BS is equipped with $M$ antennas and is serving $K$ single-antenna users. In the proposed setup, ${\vec h}^{l}_{l',k}\sim \CN({\vec{0}},\beta^{l}_{l',k} {\vec I}_{M})$ is the channel response between BS $l$ and user $k$ in cell $l'$, where $\beta^{l}_{l',k} \geq 0$ is the corresponding large-scale fading coefficient.

	We use the conventional block fading to model the randomness of the channels over time and frequency. \textcolor{black}{The coherence block of a channel is defined as the time-frequency block in which the channel response is frequency-flat and static in time}. The channels change independently from one block to another according to a stationary ergodic random process. The number of samples per coherence block is given by $\tau_c = T_{c} B_{c}$, where $T_{c}$ is the coherence time and $B_{c}$ is the coherence bandwidth \cite[Ch.~2]{redbook},\cite[Ch.~2]{bjornson2017massive}.

	Therefore, it is assumed that channel estimation is carried out at each BS once per coherence block. Each user transmits a pilot sequence from a predefined set of orthogonal pilots.  $\tau_p$ samples (with $\tau_p\leq \tau_c$) are dedicated for pilot transmission and the remaining samples will be utilized for uplink (UL) and downlink (DL) data transmission. The channel estimation phase follows the standard minimum mean square error (MMSE) estimation approach in the literature, e.g., \cite{redbook,bjornson2017massive,kay1993fundamentals} and the derivation is omitted here. The MMSE estimate of ${\vec h}^{l}_{l',k}$ is denoted as ${\hat{\vec h}}^{l}_{l',k}\sim \CN({\vec{0}},\gamma^{l}_{l',k} {\vec I}_{M})$, where $\gamma^{l}_{l',k}$ is the corresponding variance:
	\begin{equation}
	\begin{aligned}
	\gamma^{l}_{l',k}=  \frac{{\tau_p \rho_{\rm ul}\left({\beta^{l}_{l',k}}\right)^2}}{{1+ \tau_p \rho_{\rm ul}\sum\limits_{l'' \in \mathcal{P}_l}\beta^{l}_{l'',k}}},\quad l' \in \mathcal{P}_l,
	\end{aligned}
	\end{equation}
	where $\mathcal{P}_l$ is set of the BSs that are using the same $K$ pilot sequences as BS $l$. If two BSs are sharing pilots, user $k$ in the respective cells use identical pilots for $k=1,\ldots,K$.
	
	We assume that each BS performs maximum ratio processing during the data transmission phase. The detailed derivation of the SEs of both UL and DL data transmission for multi-cell massive MIMO setup is provided in  \cite{redbook,bjornson2017massive} and is therefore omitted here. The ergodic SE of user $k$ in cell $l$ is given by \cite[Th.~4.4]{bjornson2017massive}
	\begin{equation}
	\mathrm{SE}^{\rm ul}_{lk} = \left(1- \frac{\tau_p}{\tau_c}\right)\log_2\left(1+ \mathrm{SINR}^{\rm ul}_{lk}\right),
	\end{equation}

	where $\mathrm{SINR}^{\rm ul}_{lk}$, given in \eqref{eq:ul_sinr} at the top of the page,  is the effective UL SINR of user $k$ in cell $l$.  For DL data transmission, the ergodic SE of user $k$ in cell $l$ is 
	\begin{equation}
	\mathrm{SE}^{\rm dl}_{lk} = \left(1- \frac{\tau_p}{\tau_c}\right)\log_2\left(1+ \mathrm{SINR}^{\rm dl}_{lk}\right),
	\end{equation}
	where the effective DL SINR for the case of maximum ratio processing at the BSs is provided in \eqref{eq:dl_sinr} at the top of the page \cite[Ch.~4]{redbook}, $\rho_{\rm ul}$ and $\rho_{\rm dl}$ are UL and DL normalized transmit powers, respectively. In both UL and DL, $\eta_{lk} \in [0,1]$ is the PC coefficient of user $k$ in cell $l$ and these will be optimization variables in this paper.

	\section{Problem Formulation} 
	\label{probelms}
	This section motivates and defines the problem formulation. Specifically, we will propose a new multi-cell MMF PC scheme.

	In order to evaluate and compare our proposed scheme with the state-of-the-art, we also define and solve two more optimization problems. The first one is NW-MMF. Notice that MMF is the ideal utility function in a network where everyone has the same demand for data. It provides equal performance among all the users by prioritizing the user with the weakest channels. However, this scheme is not scalable and by increasing the number of cells in the network, we may end up with zero SE for all users---uniform but bad performance for everyone. It happens because the probability of having a user in deep fade due to shadow fading increases and this penalizes the whole network. More precisely, when $\beta_{l,k}^{l} \to 0$ for one user, the SE goes to zero for all the users.
	The MMF problem for UL data transmission is defined as \cite[Ch.~7]{bjornson2017massive}
	\setcounter{equation}{5}
	\begin{equation}\label{opt_problem2}
	\begin{aligned}
	& \underset{\{\eta_{lk}\}}{\text{maximize}}\quad \underset{l,k}{\text{min}}
	& &  \mathrm{SINR}_{lk}^{\rm{ul}}  \\
	& \text{subject to}
	& &  0 \leq \eta_{lk} \leq 1, \forall~l,k. \\
	\end{aligned}
	\end{equation}
	The power constraints reflect that each user has its own power amplifier and can transmit at any power $\rho_{\rm ul}\eta_{lk}$ from $0$ to $\rho_{\rm ul}$.
	Note that the optimization problem for the DL is similar to \eqref{opt_problem2} but using $\mathrm{SINR}^{\rm{dl}}_{lk}$ and different power constraints: $\sum_{k=1}^{K} \eta_{lk} \leq 1,~\forall l$ since every BS can allocate its maximum power $\rho_{\rm dl}$ freely between its users so that user $k$ in cell $l$ is allocated $\rho_{\rm dl} \eta_{lk}$.
	\subsection{Geometric-mean per-cell max-min fairness}
	To solve the scalability issue of NW-MMF, we formulate a new optimization problem in which the optimization objective is the GM of per-cell MMF of SINRs of the cells. The optimization problem for the UL data transmission is
	\vspace{-.2cm} 
	\begin{equation}\label{opt_problem1}
	\begin{aligned}
	& \underset{\{t_l\},\{\eta_{lk}\}}{\text{maximize}}
	& &  \prod_{l=1}^L \log_2\left(1+\epsilon+t_l\right)   \\
	& \text{subject to}
	& &  0 \leq \eta_{lk} \leq 1, \forall~l,k, \\
	& & & \mathrm{SINR}^{\rm ul}_{lk} \geq t_l, \forall~l,k,\\
	\end{aligned}
	\end{equation} 
	where $t_l$ is the minimum SINR of cell $l$ and $\epsilon> 0$ is a small control parameter that prevents the utility from being identically zero when one cell has a user with a very poor channel (note that zero SINR can happen in cell $l$ when $\min_{k}({\beta_{l,k}^{l}}) = 0$). We also define $1_{\epsilon}= 1+\epsilon$ to simplify the notation. The first constraint in \eqref{opt_problem1} deal with the PC coefficients for the UL data transmission of users in each cell, and the second constraint is to perform MMF on the SINR of each cell. This constraint guarantees to give the same SINR to every user within a cell, but the SINR value can be different from other cells. Therefore, a cell where all users have poor channels will not prevent the users in other cells from achieving higher SINRs. The GM utility of the per-cell SEs provides proportional fairness between cells.
	
	This optimization problem is the DL counterpart to \eqref{opt_problem1}:  
	
	\begin{equation}\label{opt_problem1DL}
	\begin{aligned}
	& \underset{\{t_l\},\{\eta_{lk}\}}{\text{maximize}}
	& &  \prod_{l=1}^L \log_2\left(1_{\epsilon}+t_l\right)   \\
	& \text{subject to}
	& &   \eta_{lk} \geq 0,~\forall~l,k, \\
	& & &  \sum_{k = 1}^{K}\eta_{lk} \leq 1,~\forall~l, \\
	& & & \mathrm{SINR}^{\rm{dl}}_{lk} \geq t_l,~\forall~l,k.\\
	\end{aligned}
	\end{equation}
	The differences from the UL are the SINR expressions being used and the power constraints, which are now reflecting the fact that each BS can distribute its power arbitrarily between its users.

	\subsection{Network-wide Proportional Fairness}
	Next, we consider the alternative network utility function with the product of the SINRs. Maximizing this objective provides NW-PF with respect to the SINRs of the users in the network. It is shown in  \cite[Sec.~7.1]{bjornson2017massive} that this objective is a lower bound on the sum SE of the network, but with greater emphasis on fairness since the utility is zero if any user gets zero SE. We can write the optimization problem for UL data transmission as 
	\begin{equation}\label{opt_problem3}
	\begin{aligned}
	& \underset{\{\textcolor{black}{t_{lk}}\},\{\eta_{lk}\}}{\text{maximize}}
	& & \prod_{l=1}^L \prod_{k=1}^K t_{lk}   \\
	& \text{subject to}
	& &  0 \leq \eta_{lk} \leq 1, \forall \textcolor{black}{l,k,} \\
	& & & \mathrm{SINR}^{\rm ul}_{lk} \geq t_{lk}, \forall l,k,\\
	\end{aligned}
	\end{equation}
	where $t_{lk}$ indicates the effective SINR of user $k$ located at cell $l$. 
The corresponding DL optimization problem is formulated as 
	\begin{equation}\label{opt_problem3DL}
	\begin{aligned}
	& \underset{\{\textcolor{black}{t_{lk}}\},\{\eta_{lk}\}}{\text{maximize}}
	& & \prod_{l=1}^L \prod_{k=1}^K t_{lk}   \\
	& \text{subject to}
	& &   \eta_{lk} \geq 0, \forall l,k \\
	& & & \sum_{k=1}^{K}\eta_{lk} \leq 1,  \forall l \\
	& & & \mathrm{SINR}^{\rm dl}_{lk} \geq t_{lk}, \forall  l,k.\\
	\end{aligned}
	\end{equation}
	The difference between this optimization problem and our new proposed formulation is that this optimization problem deals with each user individually, so there will be large SE differences within a cell.
	
	\section{Solutions to the Proposed Problems}
	\label{solution}
	In this section, we provide solutions to the optimization problems introduced in Section \ref{probelms}. 
	First, we solve the proposed GM per-cell MMF PC for the UL data transmission given in \eqref{opt_problem1}.  We can rewrite the optimization problem as 
	\vspace{-.2cm}
	\begin{equation}\label{opt_problem_re}
	\begin{aligned}
	& \underset{\{t_l\},\{\eta_{lk}\}}{\text{maximize}}
	& &  \sum_{l=1}^L \log \left( \log_2\left(1_{\epsilon}+t_l\right)\right)   \\
	& \text{subject to}
	& &  0 \leq \eta_{lk} \leq 1, \forall~l,k, \\
	& & & \mathrm{SINR}^{\rm ul}_{lk} \geq t_l, \forall~l,k,\\
	\end{aligned}
	\end{equation}
	since the natural logarithm is a monotonically increasing function. The problems are identical in terms of having the same optimal solution. Thus, the product of the SE of the cells is written as the sum of the logarithms of these SEs. By the change of variables 
	\vspace{-.2cm} 
	\begin{equation}
	\begin{aligned}
	t_l &= e^{\bar{t}_l},\quad\quad\eta_{lk} & = e^{\bar{\eta}_{lk}},
	\end{aligned}
	\end{equation}
	we obtain the equivalent reformulated problem provided in \eqref{opt_problem_re2}. 
	\begin{figure*}
		{{\begin{equation}\label{opt_problem_re2}
				\begin{aligned}
				& \underset{\{\bar{t}_l\},\{\bar{\eta}_{lk}\}}{\text{maximize}}
				& &  \sum_{l=1}^L \log\left(\log_2\left(1_{\epsilon}+e^{\bar{t}_l}\right)\right)   \\
				& \text{subject to}
				& & e^{\bar{\eta}_{lk}} \leq 1, \forall~l,k \\
				& & & \frac{M \rho_{\rm ul} \gamma^{l}_{lk} e^{\bar{\eta}_{lk}}}{1+\rho_{\rm ul} \sum\limits_{l'\in \mathcal{P}_l}^{} \textcolor{black}{\sum\limits_{k' = 1}^{K}}\beta^{l}_{l'k'}  e^{\bar{\eta}_{l'k'}} + \rho_{\rm ul}\sum\limits_{l'\notin \mathcal{P}_l}^{} \textcolor{black}{\sum\limits_{k' = 1}^{K}}\beta^{l}_{l'k'} e^{\bar{\eta}_{l'k'}}+M\rho_{\rm ul}\sum\limits_{l'\in \mathcal{P}_l\backslash\{l\}}^{}\gamma^{l}_{l'k}e^{\bar{\eta}_{l'k}}} \geq e^{\bar{t}_l}, \forall~l,k.\\
				\end{aligned}
				\end{equation}
			}
			\hrulefill
			\vspace{-4mm}
	}\end{figure*}
	We observe that the constraints can be rewritten as
	\vspace{-.2cm}
	\begin{equation}
	\begin{aligned}
	&e^{\bar{t}_l-\bar{\eta}_{lk}}+\rho_{\rm ul} \sum\limits_{l'=1}^{L}\sum\limits_{k'=1}^{K}\beta^{l}_{l'k'} e^{\bar{\eta}_{l'k'}+\bar{t}_l-\bar{\eta}_{lk}}\\
	&+M\rho_{\rm ul}\sum\limits_{l'\in \mathcal{P}_l\backslash\{l\}}^{}\gamma^{l}_{l'k}e^{\bar{\eta}_{l'k}+\bar{t}_l-\bar{\eta}_{lk}}\leq M \rho_{\rm ul} \gamma^{l}_{lk}.
	\end{aligned}
	\end{equation}
	After taking the logarithm of both sides, we have a log-sum-exponential function, which is a convex function less than or equal to a constant. This is a convex constraint. Therefore the only concern is whether the objective function in \eqref{opt_problem_re2} is concave or not. We provide the following theorem that shows that the objective function is a concave function and thus \eqref{opt_problem_re2} is a convex problem. Note that solving the optimization problem for the DL case follows the same steps as UL, hence it is omitted to avoid repetition.
	\begin{theorem}
		The function  $f(x) = \log \left(\log \left( 1_{\epsilon}+ e^{x}\right)\right)$ is a concave function with respect to $x$ for $x\geq0$ \textcolor{black}{and for any $\epsilon > 0$}.
	\end{theorem}
	\begin{proof}
		The first derivative of $f(x)$ is 
		\begin{equation}
		f'(x) = \frac{1}{\log(1_{\epsilon}+e^x)}\frac{1}{1_{\epsilon}+e^x}e^x = \frac{e^x}{1_{\epsilon}+e^x}\dfrac{1}{\log(1_{\epsilon}+e^x)}
		\end{equation}
		and the second derivative can be written as 
		\begin{equation}
		\begin{aligned}
		f''(x) &= \left(\frac{e^x}{1_{\epsilon}+e^x}\right)' \frac{1}{\log(1_{\epsilon}+e^x)} + \left(\frac{1}{\log(1_{\epsilon}+e^x)}\right)' \frac{e^x}{1_{\epsilon}+e^x}\\
		&= \frac{\left(1_{\epsilon}+e^x\right)\left(e^x\right)' - e^x\left(1_{\epsilon}+e^x\right)'}{\left(1_{\epsilon}+e^x\right)^2}\frac{1}{\log(1_{\epsilon}+e^x)} \\
		&+ \frac{-\frac{e^x}{1_{\epsilon}+e^x}}{\left(\log(1_{\epsilon}+e^x)\right)^2}\frac{e^x}{1_{\epsilon}+e^x}\\
		&= \frac{e^x + (e^x)^2 - (e^x)^2}{\left(1_{\epsilon}+e^x\right)^2 \log(1_{\epsilon}+e^x)} - \frac{(e^x)^2}{\left(1_{\epsilon}+e^x\right)^2 (\log(1_{\epsilon}+e^x))^2}\\
		&= \left(1-\frac{e^x}{\log(1_{\epsilon}+e^x)}\right)\frac{e^x}{(1_{\epsilon}+e^x)^2(\log(1_{\epsilon}+e^x))^2}.
		\end{aligned}
		\end{equation}
		Now we define $g(x)= e^x -\log(1+e^x)$ and we have 
		\begin{equation}
		\begin{aligned}
		g'(x) &= e^x - \frac{e^x}{1_{\epsilon}+e^x} = e^x \left(1-\frac{1}{1_{\epsilon}+e^x}\right) \geq 0.
		\end{aligned}
		\end{equation}
		This means that $g(x)$ is monotonically increasing in $x$, we also have $g(0)>0$. Therefore we have shown that $g(x)>0, \forall x\geq 0$. 
		This implies $f''(x)\leq 0, \forall x\geq 0$. Therefore, we have proved that $f(x)=\log(\log(1_{\epsilon}+e^x))$ is a concave function in $x$. 
	\end{proof}
	As the objective is a sum of concave functions, it is also jointly concave. Hence, we have shown that \eqref{opt_problem1} is a convex problem and it follows that every stationary point is also a global optimum solution. Therefore any algorithm that converges to a stationary point can be applied to solve the problem. In the simulation part, we use the interior point algorithm within the \texttt{fmincon} solver in MATLAB.

	To solve the NW-MMF optimization problems given in \eqref{opt_problem2} and \eqref{opt_problem3} one can write it on epigraph form and solve linear feasibility optimization problems using the bisection algorithm. The details of the bisection algorithm can be found in \cite[Ch.~7]{bjornson2017massive}.

	The NW-PF problem for both UL and DL data transmission, given in \eqref{opt_problem3} and \eqref{opt_problem3DL}, are geometric problems \cite{boyd2007tutorial,gpPowercont}. The detailed proof is provided in \cite[Th.~7.2]{bjornson2017massive}. These optimization problems can be solved efficiently by using standard convex optimization solvers, for example, we used CVX \cite{grant2008cvx} in the simulation part.
	\begin{figure}[htb!] \vspace{-2mm}
		\begin{minipage}[t]{0.43\linewidth}
			\centering
			\centerline{\includegraphics[width=5cm]{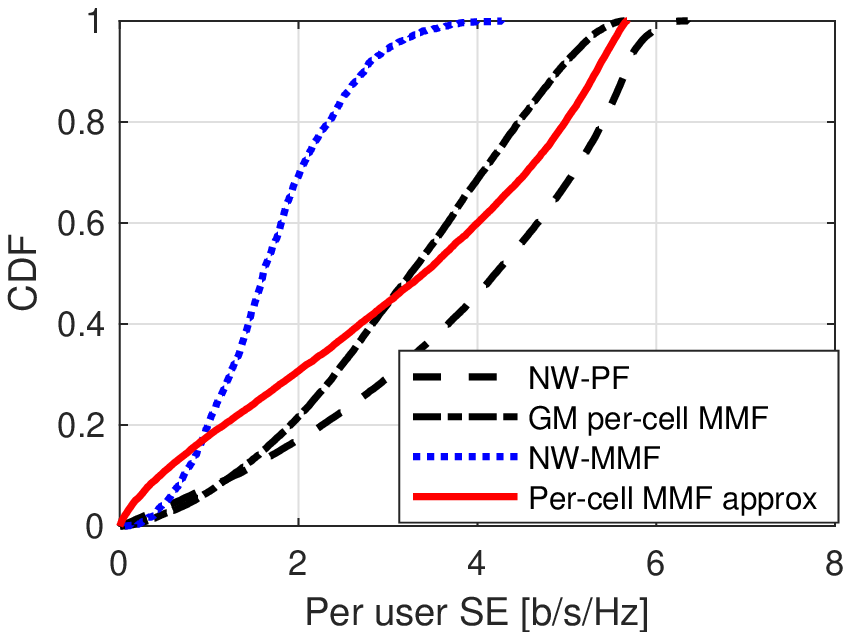}}
			\centerline{(a) SE of CU  $k$}\medskip
		\end{minipage}
		\hfill
		\begin{minipage}[t]{0.43\linewidth}
			\centering
			\centerline{\includegraphics[width=5cm]{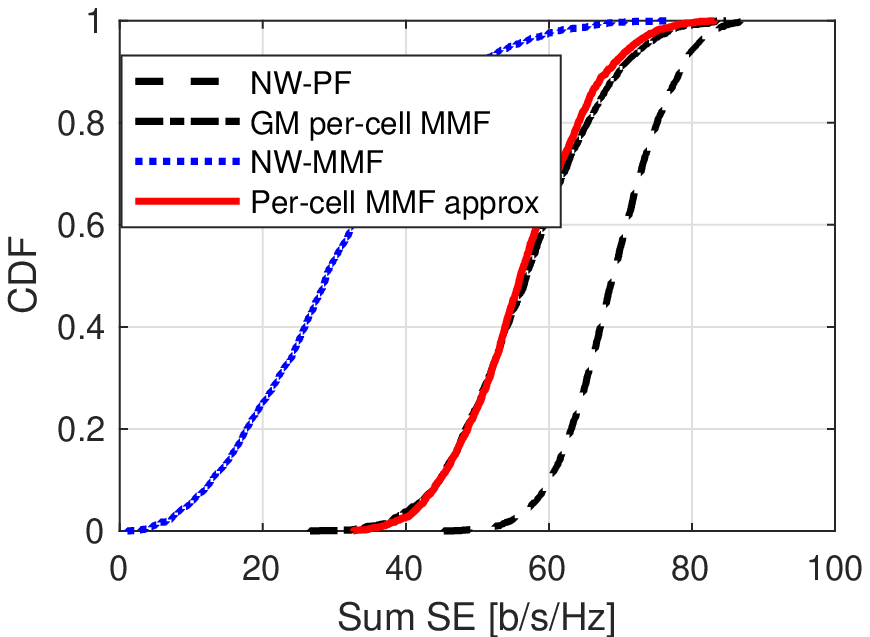}}
			\centerline{(b) Sum SE}\medskip
		\end{minipage} \vspace{-2mm}
		\caption{Uplink data transmission.}
		\label{fig:ul} \vspace{-2mm}
	\end{figure}
	\section{Numerical Analysis}
	In this section, we provide a numerical comparison of the three PC algorithms provided in Section \ref{probelms}. In addition, we consider the heuristic per-cell max-min algorithm proposed in \cite[Ch.~6]{redbook}. We consider a multi-cell massive MIMO setup consisting of $9$ cells and we use wrap-around to avoid edge effects. \textcolor{black}{Each BS is equipped with $M = 100$ antennas.} We assume a square grid layout where each square has a BS in the center and all of the BSs are located in a $1$\,km$^2$ area. Furthermore, each BS serves two users that are randomly distributed with uniform distribution in the coverage area of the BS. We assume a reuse factor of one, which means all the BSs are in the set $\mathcal{P}_l$. The bandwidth is $20\,$MHz and each coherence block contains $200$ symbols. The large-scale fading coefficients are modeled as \cite{bjornson2017massive}
	\begin{equation}
	\beta^{l}_{l',k} \left[{\rm dB}\right] = -35 - 36.7\log_{10}\left(d^{l}_{l',k}/1\,\rm{m}\right) +F^{l}_{l',k},
	\end{equation}
	where $d^{l}_{l',k} $ is distance between user $k$ located in cell $l'$ to BS $l$. In addition, $F^{l}_{l',k}$ is
	shadow fading generated from a log-normal
	distribution with standard deviation 8 ${\rm dB}\,$. The noise variance is set to~$-94\,$dBm and the maximum transmit power of the users is~$200\,$mW for UL data transmission. The maximum transmit power of the BS is selected to be $40\,$W. The simulations consider $2000$ realizations, where the users are dropped randomly in each cell. Figs.\ref{fig:ul}(a) and  \ref{fig:dl}(a) plot the cumulative distribution function (CDF) of the SE of all the users for UL and DL data transmission, respectively.
	
	In these figures, it can be seen that NW-PF gives higher SE than the proposed GM of per-cell MMF for most users but not in the lower tail which is the important part for delivering fairness and uniform performance. 
The 	$10\%$ weakest users in the UL and the $12\%$ weakest users in the DL get higher SE when using the proposed GM per-cell MMF.
The proposed scheme is comparable to NW-MMF in the lower tail (i.e., the weakest cell in the network) but substantially better for all other cells.

	\begin{figure}[htb] \vspace{-2mm}
		\begin{minipage}[t]{1\linewidth}
			\centering
			\centerline{\includegraphics[width=5.5cm]{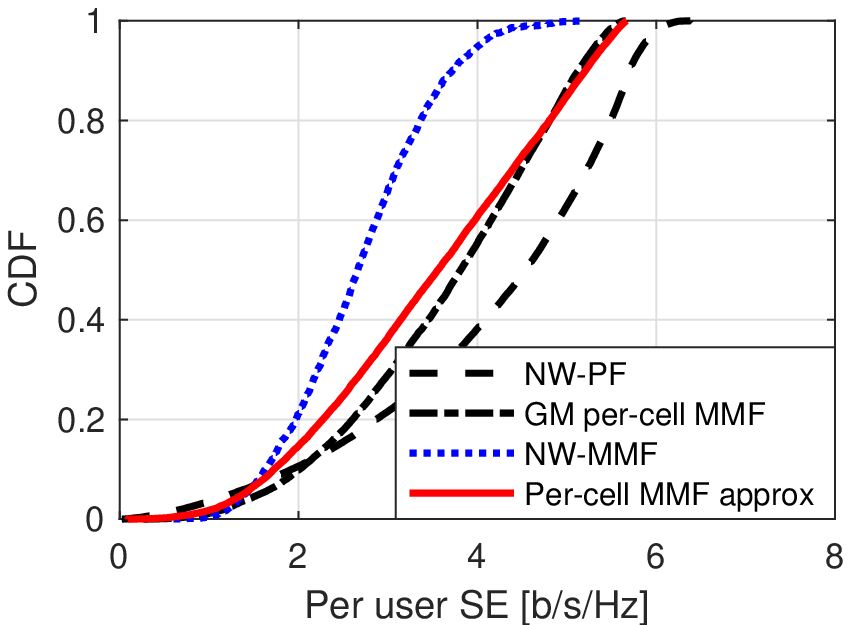}}
			\centerline{(a) SE of CU  $k$}\medskip
		\end{minipage}
		\begin{minipage}[t]{1\linewidth}
			\centering
			\centerline{\includegraphics[width=5.5cm]{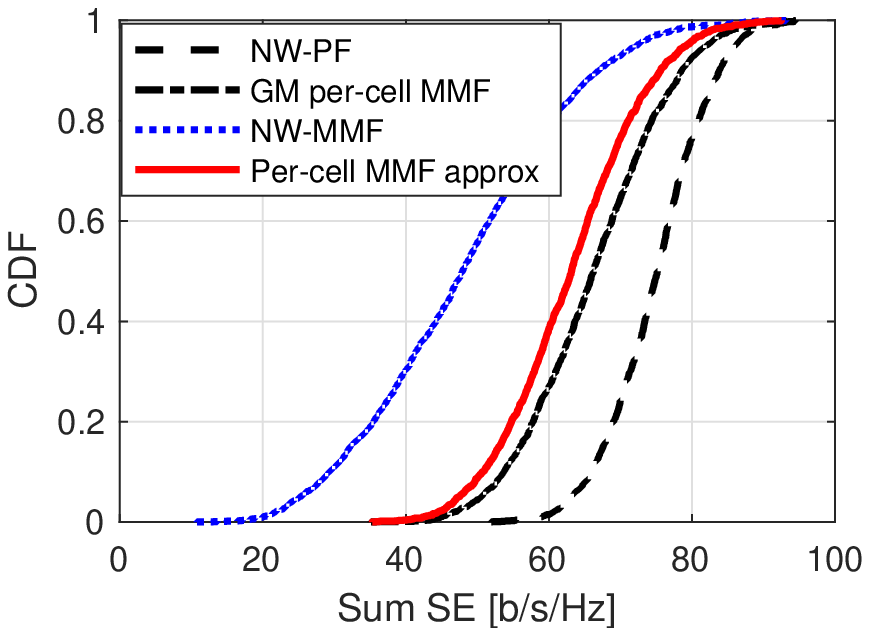}}
			\centerline{(b) Sum SE}\medskip
		\end{minipage}
		\vspace{-4mm}
		\caption{Downlink data transmission.} 		\vspace{-2mm}
		\label{fig:dl}	
	\end{figure}
	If we do a one-to-one comparison for all the users and calculate the percentage of users that get better SE with NW-MMF than with NW-PF or the proposed scheme, we get the results provided in {Table}~\ref{table:1}. These numbers show that one fifth of the users get higher SE, but we also see from the CDF curves that their SE gains are tiny, thus percentage values like this need to be taken with a grain of salt.
	\begin{table}[b!]
		\centering \vspace{-4mm}
		\caption{Percentage of users getting better SE using NW-MMF.} \vspace{+1mm}
		\begin{center}
			\begin{tabular}{| p{1.1cm}| p{1.5cm}| p{2.5cm}|}
				\hline
				&NW-PF& GM per-cell MMF \\ 
				\hline
				Uplink & 14\% & 17\% \\
				\hline
				Downlink & 18\% & 22\% \\
				\hline
			\end{tabular}\label{table:1}\vspace{-0.5cm}
		\end{center} 
	\end{table} 

	Another observation is that, in both figures, the heuristic scheme from \cite[Ch.~6]{redbook} provides similar performance as our proposed method in the DL, but the proposed method generally gives higher SE. In Fig.~\ref{fig:ul}(b) and  \ref{fig:dl}(b), we plot the CDF of the sum SE of the whole network. We see that the NW-PF scheme performs the best in terms of sum SE as it can be seen as an approximation to the sum SE maximization in the high SINR regime. 
	\vspace{-4mm}
	
	\section{Conclusion}
	In this paper, we analyzed different power control schemes that target fairness in multi-cell massive MIMO systems. 
	We proposed to maximize the geometric mean of the per-cell max-min SEs. This approach is not subject to the same scalability issues as the conventional NW-MMF approach, which has received much attention in the literature. We solved the new problem formulation to global optimality and achieved better or comparable performance as the previous heuristic scheme in \cite{redbook} that also targeted to resolve the scalability issue of NW-MMF. Furthermore, our proposed approach provides more fairness towards weak users in comparison with NW-PF.

	\bibliographystyle{IEEEtran}
	\bibliography{di}
	
\end{document}

%% file: commands.tex
\newcommand{\be}{\begin{equation}}
\newcommand{\ee}{\end{equation}}
\newcommand{\bal}{\begin{aligned}}
\newcommand{\eal}{\end{aligned}}	
\newcommand{\ba}{\begin{array}}
\newcommand{\ea}{\end{array}}
\newcommand{\bea}{\begin{eqnarray}}
\newcommand{\eea}{\end{eqnarray}}

\newcommand{\vbar}{\raisebox{.17ex}{\rule{.04em}{1.35ex}}}
\newcommand{\vbarind}{\raisebox{.01ex}{\rule{.04em}{1.1ex}}}
\newcommand{\R}{\ifmmode {\rm I}\hspace{-.2em}{\rm R} \else ${\rm I}\hspace{-.2em}{\rm R}$ \fi}
\newcommand{\T}{\ifmmode {\rm I}\hspace{-.2em}{\rm T} \else ${\rm I}\hspace{-.2em}{\rm T}$ \fi}
\newcommand{\N}{\ifmmode {\rm I}\hspace{-.2em}{\rm N} \else \mbox{${\rm I}\hspace{-.2em}{\rm N}$} \fi}
\newcommand{\B}{\ifmmode {\rm I}\hspace{-.2em}{\rm B} \else \mbox{${\rm I}\hspace{-.2em}{\rm B}$} \fi}
\newcommand{\Hil}{\ifmmode {\rm I}\hspace{-.2em}{\rm H} \else \mbox{${\rm I}\hspace{-.2em}{\rm H}$} \fi}
\newcommand{\C}{\ifmmode \hspace{.2em}\vbar\hspace{-.31em}{\rm C} \else \mbox{$\hspace{.2em}\vbar\hspace{-.31em}{\rm C}$} \fi}
\newcommand{\Cind}{\ifmmode \hspace{.2em}\vbarind\hspace{-.25em}{\rm C} \else \mbox{$\hspace{.2em}\vbarind\hspace{-.25em}{\rm C}$} \fi}
\newcommand{\Q}{\ifmmode \hspace{.2em}\vbar\hspace{-.31em}{\rm Q} \else \mbox{$\hspace{.2em}\vbar\hspace{-.31em}{\rm Q}$} \fi}
\newcommand{\Z}{\ifmmode {\rm Z}\hspace{-.28em}{\rm Z} \else ${\rm Z}\hspace{-.28em}{\rm Z}$ \fi}

\renewcommand{\vec}[1]{\bf{#1}}     

\newcommand{\CN}{\mathcal{CN}}

\renewcommand{\theenumii}{.\arabic{enumii}}